\documentclass{article}

\usepackage{amsmath,amsgen,latexsym}
\usepackage{amstext,amssymb,amsfonts,latexsym}
\usepackage{theorem}
\usepackage[dvips]{graphics,epsfig}
\usepackage[dvips]{graphicx}
\usepackage{hyperref}

 \newcommand{\hs}[1]{\hspace*{ #1 mm}}

 \newcommand{\real}{\mathbb{R}}

 \newcommand{\nat}{\mathbb{N}}
 
 \newcommand{\integer}{\mathbb{Z}}

 \newcommand{\poly}{\mathrm{poly}}

 \def\bbox{\vrule height6pt width6pt depth1pt}

\theoremstyle{plain}
\theoremheaderfont{\bfseries}
 \newtheorem{theorem}{Theorem}[section]
 \newtheorem{lemma}[theorem]{Lemma}

 {\theorembodyfont{\rmfamily}
  \newtheorem{definition}[theorem]{Definition}}
 {\theorembodyfont{\rmfamily} }
 {\theorembodyfont{\rmfamily} }

 \newenvironment{proof}{\par \noindent
            {\bf Proof. \hs{2}}}{\hfill\bbox \vspace*{3mm}}

 \newcommand{\bra}[1]{\langle #1 |}
 \newcommand{\ket}[1]{| #1 \rangle}
 \newcommand{\braket}[2]{\langle #1 , #2 \rangle}

\newif\ifnotesw\noteswtrue
   {\ifnotesw\marginpar[\hfill\(\top\)]{\(\top\)}\fi}%
      {\ifnotesw\marginpar[\hfill\(\bot\)]{\(\bot\)}\fi}
      
\newcommand{\mnote}[1]%
   {\ifnotesw\marginpar%
	  [{\scriptsize\begin{minipage}[t]{\marginparwidth}
	  \raggedleft#1%
		  \end{minipage}}]%
	  {\scriptsize\begin{minipage}[t]{\marginparwidth}
	  \raggedright#1%
		  \end{minipage}}%
    \fi}

\newcommand{\ignore}[1]{}

\usepackage{color}
\usepackage{hyperref}
\usepackage{comment}

\usepackage{soulutf8} 
\usepackage{todonotes} 

\begin{document}

\author{Benjam\'in Bar\'an\thanks{Email: \texttt{bbaran@cba.com.py}} \qquad\qquad Marcos Villagra\thanks{Email: \texttt{mvillagra@pol.una.py}; corresponding author}\\
\\
\emph{Universidad Nacional de Asunci\'on}\\
\emph{N\'ucleo de Investigaci\'on y Desarrollo Tecnol\'ogico (NIDTEC)}
}
\title{Multiobjective Optimization in a Quantum Adiabatic Computer\thanks{An extended abstract of this paper appeared in Ref.~\cite{BV16}}}
\date{}
\maketitle

\begin{abstract}
In this work we present a quantum algorithm for multiobjective combinatorial optimization. We show how to map a convex combination of objective functions onto a Hamiltonian and then use that Hamiltonian to prove that the quantum adiabatic  algorithm of Farhi \emph{et al.}\cite{FGG00} can find Pareto-optimal solutions in finite time provided certain convex combinations of objectives are used and the underlying multiobjective problem meets certain restrictions.
\end{abstract}
\noindent{{\bf Keywords:}} quantum computation, multiobjective optimization, quantum adiabatic evolution.\\

\section{Introduction}
Optimization problems are pervasive in everyday applications like logistics, communication networks, artificial intelligence and many other areas. Consequently, there is a high demand of efficient algorithms for these problems. Many algorithmic and engineering techniques applied to optimization problems are being developed to make an efficient use of computational resources in optimization problems. In fact, several engineering applications are multiobjective optimization problems, where several objectives must be optimized at the same time. For a survey on multiobjective optimization see for example Refs.~\cite{EG00,LBB14}. In this work, we present what we consider the first algorithm for multiobjective optimization using a quantum adiabatic computer.

Quantum computation is a promising paradigm for the design of highly efficient algorithms based on the principles of quantum mechanics. Researchers have studied the computational power of quantum computers by showing the advantages it presents over classical computers in many applications. Two of the most well-know applications are in unstructured search and the factoring of composite numbers. In structured search, Grover's algorithm can find a single marked element among $n$ elements in time $O(\sqrt{n})$, whereas any other classical algorithm requires time at least $n$ \cite{Gro96}.  Shor's algorithm can factor composite numbers in polynomial time---any other known classical algorithm can find factors of composite numbers in subexponential time (it is open whether a classical algorithm can find factors in polynomial time) \cite{Sho94}.

Initially, before the year 2000, optimization problems were not easy to construct using quantum computers. This was because most studied models of quantum computers were based on quantum circuits which presented difficulties for the design of optimization algorithms. The first paper reporting on solving an optimization problem was in Ref.~\cite{DH99}. Their algorithm finds a minimum inside an array of $n$ numbers in time $O(\sqrt{n})$. More recently, Baritompa \emph{et al.}\cite{BBW05} presented an improved algorithm based on Ref.~\cite{DH99}; this latter algorithm, however, does not have a proof of convergence in finite time. The algorithms of Refs.~\cite{DH99} and \cite{BBW05} are based on Grover's search, and hence, in the quantum circuit model.

\emph{Quantum Adiabatic Computing} was introduced by Farhi \emph{et al.}\cite{FGG00} as a new quantum algorithm and computation paradigm more friendly to optimization problems. This new paradigm is based on a natural phenomenon of quantum annealing \cite{DC08}; analogously to classical annealing, optimization problems are mapped onto a natural optimization phenomenon, and thus, optimal solutions are found by just letting this phenomenon to take place.

The algorithms of Refs. \cite{DH99} and \cite{BBW05} are difficult to extend to multiobjective optimization and to prove convergence in finite time. Hence, quantum adiabatic computing presents itself as a more suitable model to achieve the following two goals: (i) to propose a quantum algorithm for multiobjective optimization and (ii) prove convergence in finite time of the algorithm.

In this work, as our main contribution, we show that the quantum adiabatic algorithm of Farhi \emph{et al.}\cite{FGG00} can be used to find Pareto-optimal solutions in finite time provided certain restrictions are met. In Theorem \ref{the:alg}, we identify two structural features that any multiobjective optimization problem must have in order to use the abovementioned adiabatic algorithm.

The outline of this paper is the following. In Section \ref{sec:MCO} we present a brief overview of multiobjective combinatorial optimization and introduce the notation used throughout this work; in particular, several new properties of multiobjective combinatorial optimization are also presented that are of independent interest. In Section \ref{sec:adiabatic} we explain the quantum adiabatic theorem, which is the basis of the adiabatic algorithm. In Section \ref{sec:ada-algorithm} we explain the adiabatic algorithm and its application to combinatorial multiobjective optimization. In Section  \ref{sec:convex} we prove our main result of Theorem \ref{the:alg}. In Section \ref{sec:two-parabolas} we show how to use the adiabatic algorithm in a concrete problem. Finally, in Section \ref{sec:open} we present a list of challenging open problems.

\section{Multiobjective Combinatorial Optimization}\label{sec:MCO}
In this section we introduce the notation used throughout this paper and the main concepts of multiobjective optimization. The set of natural numbers (including 0) is denoted $\nat$, the set of integers is $\integer$, the set of real numbers is denoted $\real$ and the set of positive real numbers is $\real^+$. For any $i,j\in \nat$, with $i< j$, we let $[i,j]_\integer$ denote the discrete interval $\{i,i+1,\dots, j-1,j\}$. The set of binary words of length $n$ is denoted $\{0,1\}^n$. We also let $\poly(n)=O(n^c)$ be a polynomial in $n$.


A \emph{multiobjective combinatorial optimization problem} (or MCO) is an optimization problem involving multiple objectives over a finite set of feasible solutions. These objectives typically present trade-offs among solutions and in general there is no single optimal solution. In this work, we follow the definition of Ref. \cite{KLP75}. Furthermore, with no loss of generality, all optimization problems considered in this work are minimization problems.

Let $S_1,\dots,S_d$ be totally ordered sets and let $\leq_i$ be an order on set $S_i$ for each $i\in [1,d]_{\integer}$. We also let $n_i$ be the cardinality of $S_i$. Define the natural partial order relation $\prec$ over the cartesian product $S_1 \times \cdots \times S_d$ in the following way. For any $u=(u_1,\dots,u_d)$ and $v=(v_1,\dots,v_d)$ in $S_1\times \cdots \times S_d$, we write $u\prec v$ if and only if for any $i\in[1,d]_\integer$ it holds that $u_i \leq_i v_i$; otherwise we write $u\nprec v$. An element $u\in S$ is a \emph{minimal element} if there is no $v\in S$ such that $v\prec u$ and $v\neq u$. Moreover, we say that $u$ is \emph{non-comparable} with $v$ if $u\nprec v$ and $v\nprec u$ and succinctly write $u\sim v$. In the context of multiobjective optimization, the relation $\prec$ as defined here is often referred to as the \emph{Pareto-order} relation \cite{KLP75}.

\begin{definition}
A \emph{multiobjective combinatorial optimization problem} (or shortly, MCO) is defined as a tuple $\Pi=(D,R,d,\mathcal{F},\prec)$ where $D$ is a finite set called domain, $R\subseteq \real^+$ is a set of values, $d$ is a positive integer, $\mathcal{F}$ is a finite collection of functions $\{f_i\}_{i\in[1,d]_\integer}$ where each $f_i$ maps from $D$ to $R$, and $\prec$ is the Pareto-order relation on $R^d$ (here $R^d$ is the $d$-fold cartesian product on $R$). Define a function $f$ that maps $D$ to $R^d$ as $f(x)=(f_1(x),\dots,f_d(x))$ referred as the \emph{objective vector} of $\Pi$. If $f(x)$ is a minimal element of $R^d$ we say that $x$ is a \emph{Pareto-optimal solution} of $\Pi$. For any two elements $x,y\in D$, if $f(x)\prec f(y)$ we write $x\prec y$; similarly, if $f(x)\sim f(y)$ we write $x\sim y$. For any $x,y\in D$, if $x\prec y$ and $y\prec x$ we say that $x$ and $y$ are \emph{equivalent} and write $x \equiv y$. The set of all Pareto-optimal solutions of $\Pi$ is denoted $P(\Pi)$.
\end{definition}

A canonical example in multiobjective optimization is the Two-Parabolas problem. In this problem we have two objective functions defined by two parabolas that intersect in a single point, see Fig.\ref{fig:two-parabolas}. In this work, we will only be concerned with a combinatorial version of the Two-Parabolas problem where each objective function only takes values on a finite set of numbers. 

\begin{figure*}[t]
\centering
\includegraphics[scale=0.7]{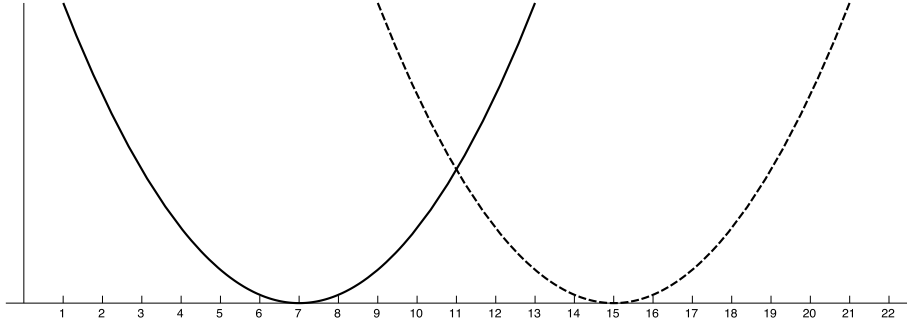}
\caption{The Two-Parabolas Problem. The first objective function $f_1$ is represented by the bold line and the second objective function $f_2$ by the dashed line. For MCOs, each objective function takes values only on the natural numbers. Note that there are no equivalent elements in the domain. In this particular example, all the solutions between 7 and 15 are Pareto-optimal.}
\label{fig:two-parabolas}
\end{figure*}

Considering that the set of Pareto-optimal solutions can be very large, we are mostly concerned on finding a subset of the Pareto-optimal solutions. Optimal query algorithms to find all Pareto-optimal solutions for $d=2,3$ and almost tight upper and lower bounds for any $d\geq 4$ up to polylogarithmic factors were discovered by \cite{KLP75}; \cite{PY00} showed how to find an approximation to all Pareto-optimal solutions in polynomial time.

For the remaining of this work, $\prec$ will always be the Pareto-order relation and will be omitted from the definition of any MCO. Furthermore, for convenience, we will often write $\Pi_d=(D,R,\mathcal{F})$ as a short-hand for $\Pi=(D,R,d,\mathcal{F})$. In addition, we will assume for this work that each function $f_i\in \mathcal{F}$ is computable in polynomial time and each $f_i(x)$ is bounded by a polynomial in the number of bits of $x$.


\begin{definition}
An MCO $\Pi_d$ is \emph{well-formed} if for each $f_i\in \mathcal{F}$ there is a unique $x\in D$ such that $f_i(x)=0$. An MCO $\Pi_d$ is \emph{normal} if it is well-formed and $f_i(x)=0$ and $f_j(y)=0$, for $i\neq j$, implies $x\neq y$.
\end{definition}

In a normal MCO, the value of an optimal solution in each $f_i$ is 0, and all optimal solutions are different. In Fig.\ref{fig:two-parabolas}, solutions 7 and 15 are optimal solutions of  $f_1$ and $f_2$ with value 0, respectively; hence, the Two-Parabolas problem of Fig.\ref{fig:two-parabolas} is normal.

\begin{definition}
An MCO $\Pi_d$ is \emph{collision-free} if given $\lambda=(\lambda_1,\dots,\lambda_d)$, with each $\lambda_i \in \real^+$, for any $i\in[1,d]_\integer$ and any pair $x,y\in D$ it holds that $|f_i(x)-f_i(y)|>\lambda_i$. If $\Pi_d$ is collision-free we write succinctly as $\Pi_d^\lambda$.
\end{definition}

The Two-Parabolas problem of Fig.\ref{fig:two-parabolas} is not collision-free; for example, for solutions 5 and 9 we have that $f_1(5)=f_1(9)$. In Section \ref{sec:two-parabolas} we show how to turn the Two-Parabolas problem into a collision-free MCO.

\begin{definition}
A Pareto-optimal solution $x$ is \emph{trivial} if $x$ is an optimal solution of some $f_i\in \mathcal{F}$.
\end{definition}

In Fig.\ref{fig:two-parabolas}, solutions 7 and 15 are trivial Pareto-optimal solutions, whereas any $x$ between 7 and 15 is non-trivial.

\begin{lemma}
For any normal MCO $\Pi_d$, if $x$ and $y$ are trivial Pareto-optimal solutions of $\Pi_d$, then $x$ and $y$ are not equivalent.
\end{lemma}
\begin{proof}
Let $x,y$ be two trivial Pareto-optimal solutions of $\Pi_d$. There exists $i,j$ such that $f_i(x)$ and $f_j(y)=0$. Since $\Pi_d$ is normal we have that $x\neq y$ and $f_i(y)>0$ and $f_j(x)>0$, hence, $x\sim y$ and they are not equivalent.
\end{proof}

Let $W_d$ be a set of of normalized vectors in $[0,1)^d$, the continuous interval between 0 and less than 1, defined as
\begin{equation}\label{eq:wd}
W_d=\left\{w=(w_1,\dots,w_d)\in [0,1)^d \Bigg|\sum_{i=1}^d w_i=1\right\}.
\end{equation}
For any $w\in W_d$, define $\langle f(x),w \rangle=\langle w,f(x)\rangle=w_1f_1(x)+\cdots + w_d f_d(x)$. 

\begin{lemma}\label{lem:cont}
Given $\Pi_d=(D,R,\mathcal{F})$, any two elements $x,y\in D$ are equivalent if and only if for all $w\in W_d$ it holds that $\langle f(x),w\rangle=\langle f(y),w\rangle$.
\end{lemma}
\begin{proof}
Assume that $x\equiv y$. Hence $f(x)=f(y)$. If we pick any $w\in W_d$ we have that
\[
\langle f(x),w\rangle = w_1f_1(x)+\cdots+w_df_d(x)=w_1f_1(y)+\cdots+w_df_d(y)=\langle f(y),w\rangle.
\]

Now suppose that for all $w\in W_d$ it holds $\langle f(x),w\rangle=\langle f(y),w\rangle$. By contradiction, assume that $x\not\equiv y$. With no loss of generality, assume further that there is exactly one $i\in [1,d]_\integer$ such that $f_i(x)\neq f_i(y)$. Hence
\begin{equation}\label{eq:cont}
w_i(f_i(x)-f_i(y))=\sum_{j\neq i}w_j(f_j(y)-f_j(x)).
\end{equation}
The right hand of Eq.(\ref{eq:cont}) is 0 because for all $j\neq i$ we have that $f_j(x)=f_j(y)$. The left hand of Eq.(\ref{eq:cont}), however, is not 0 by our assumption, hence, a contradiction. Therefore, it must be that $x$ and $y$ are equivalent.
\end{proof}

\begin{lemma}\label{lem:sup}
Let $\Pi_d=(D,R,\mathcal{F})$. For any $w\in W_d$ there exists $x\in D$ such that if $\langle f(x),w\rangle=\min_{y\in D}\{\langle f(y),w\rangle\}$, then $x$ is a Pareto-optimal solution of $\Pi_d$.
\end{lemma}
\begin{proof}
Fix $w\in W_d$ and let $x\in D$ be such that $\langle f(x),w \rangle$ is minimum among all elements of $D$.
For any $y\in D$, with $y\neq x$, we need to consider two cases: (1) $\langle f(y),w\rangle=\langle f(x),w\rangle$ and (2) $\langle f(y),w\rangle > \langle f(x),w\rangle$.
\vspace{0.2cm}

\noindent\emph{Case (1)}.
Here we have another two subcases, either $f_i(y)=f_i(x)$ for all $i$ or there exists at least one pair $i,j\in\{1,\dots,d\}$ such that $w_if_i(x)<w_if_i(y)$ and $w_jf_j(y)<w_jf_j(x)$. When $f_i(x)=f_i(y)$ for each $i=1,\dots,d$ we have that $x$ and $y$ are equivalent. On the contrary, if  $w_if_i(x)<w_if_i(y)$ and $w_jf_j(y)<w_jf_j(x)$, we have that $f_i(x)<f_i(y)$ and $f_j(y)<f_j(x)$, and hence, $x\sim y$.
\vspace{0.2cm}

\noindent\emph{Case (2)}. In this case, there exists $i\in\{1,\dots,d\}$ such that $w_if_i(x)<w_if_i(y)$, and hence, $f_i(x)<f_i(y)$. Thus, $f(y)\not\prec f(x)$ and $y\not\prec x$ for any $y\neq x$.
\vspace{0.2cm}

We conclude from Case (1)  that $x\equiv y$ or $x\sim y$, and from Case (2) that $y\nprec x$. Therefore, $x$ is Pareto-optimal.\end{proof}

In this work, we will concentrate on finding non-trivial Pareto-optimal solutions. Finding trivial elements can be done by letting $w_i=1$ for some $i\in[1,d]_\integer$ and then running and optimization algorithm for $f_i$; consequently, in Eq.(\ref{eq:wd}) we do not allow for any $w_i$ to be 1. The process of mapping several objectives to a single-objective optimization problem is sometimes referred as a \emph{linearization} of the MCO \cite{EG00}.

From Lemma \ref{lem:sup}, we know that some Pareto-optimal solutions may not be optimal solutions for any linearization $w\in W_d$. We define the set of \emph{non-supported Pareto-optimal solutions} as the set $N(\Pi)$ of all Pareto-optimal solutions $x$ such that $\langle f(x),w\rangle$ is not optimal for any $w\in W_d$. We also define the set $S(\Pi)$ of \emph{supported Pareto-optimal solutions} as the set $S(\Pi)=P(\Pi)\setminus N(\Pi)$ \cite{EG00}. 

Note that there may be Pareto-optimal solutions $x$ and $y$ that are non-comparable and $\langle f(x),w\rangle=\langle f(y),w\rangle$ for some $w\in W_d$. That is equivalent to say that the objective function obtained from a linearization of an MCO is not injective.

\begin{definition}
Any two elements $x,y\in D$ are \emph{weakly-equivalent} if and only if there exists $w\in W_d$ such that $\langle f(x),w\rangle=\langle f(y),w\rangle$.
\end{definition}

By Lemma \ref{lem:cont}, any two equivalent solutions $x,y$ are also weakly-equivalent, ; the other way, however, does not hold in general. For example, consider two objective vectors $f(x)=(1,2,3)$ and $f(y)=(1,3,2)$. Clearly, $x$ and $y$ are not equivalent; however, if $w=(1/3,1/3,1/3)$ we can see that $x$ and $y$ are indeed weakly-equivalent. In Fig.\ref{fig:two-parabolas}, points 10 and 12 are weakly-equivalent.

\section{Quantum Adiabatic Computation}\label{sec:adiabatic}
Starting from this section we assume basic knowledge of quantum computation. For a thorough treatment of quantum information science we refer the reader to the book by Nielsen and Chuang\cite{NC00}.

Let $\mathcal{H}$ be a Hilbert space with a finite basis $\{\ket{u_i}\}_{i}$. For any vector $\ket{v}=\sum_{i} \alpha_i \ket{u_i}$, the $\ell_2$-norm of $\ket{v}$ is defined as $\| v\| = \sqrt{\sum_i|\alpha_i|^2}$. For any matrix $A$ acting on $\mathcal{H}$, we define the operator norm of $A$ induced by the $\ell_2$-norm as $\|A\|=\max_{\|v\|=1}\| A\ket{v}\|$.

The Hamiltonian of a quantum system gives a complete description of its time evolution, which is governed by the well-known Schr\" odinger's equation
\begin{equation}\label{eq:schrodinger}
i\hbar \frac{d}{dt} \ket{\Psi(t)}=H(t)\ket{\Psi(t)},
\end{equation}
where $H$ is a Hamiltonian, $\ket{\Psi(t)}$ is the state of the system at time $t$, Planck's constant is denoted by $\hbar$ and $i=\sqrt{-1}$. For simplicity, we will omit $\hbar$ and $i$ from now on. If $H$ is time-independent, it is easy to see that a solution to Eq.(\ref{eq:schrodinger}) is simply $\ket{\Psi(t)}=U(t)\ket{\Psi(0)}$ where $U(t)=e^{-itH}$ using $\ket{\Psi(0)}$ as a given initial condition. When the Hamiltonian depends on time, however, Eq.(\ref{eq:schrodinger}) is not in general easy to solve and much research is devoted to it; nevertheless, there are a few known special cases.

Say that a closed quantum system is described by a time-dependent Hamiltonian $H(t)$. If $\ket{\Psi(t)}$ is the minimum energy eigenstate of $H(t)$, adiabatic time evolution keeps the system in its lower energy eigenstate as long as the change rate of the Hamiltonian is ``slow enough.'' This natural phenomenon is formalized in the \emph{Adiabatic Theorem}, first proved in Ref. \cite{BF28}. Different proofs where given along the years, see for example Refs. \cite{Kat50,Mes62,SWL04,Rei04,AR04}. In this work we make use of a version of the theorem presented in Ref. \cite{AR04}.

Consider a time-dependent Hamiltonian $H(s)$, for $0\leq s \leq 1$, where $s=t/T$ so that $T$ controls the rate of change of $H$ for $t\in [0,T]$. We denote by $H'$ and $H''$ the first and second derivatives of $H$.
\begin{theorem}[Adiabatic Theorem \cite{BF28,Kat50,AR04}]\label{the:adiabatic}
Let $H(s)$ be a nondegenerate Hamiltonian, let $\ket{\psi(s)}$ be one of its eigenvectors and $\gamma(s)$ the corresponding eigenvalue. For any $\lambda\in \real^+$ and $s\in[0,1]$, assume that for any other eigenvalue $\hat{\gamma}(s)$ it holds that $|\gamma(s)-\hat{\gamma}(s)|> \lambda$. Consider the evolution given by $H$ on initial condition $\ket{\psi(0)}$ for time $T$ and let $\ket{\phi}$ be the state of the system at $T$. For any nonnegative $\delta\in \real$, if $T\geq \frac{10^5}{\delta^2}.\max\{\frac{\|H'\|^3}{\lambda^4},\frac{\|H'\|\cdot \|H''\|}{\lambda^3}\}$ then $\|\phi-\psi(1)\| \leq \delta$.
\end{theorem}

The Adiabatic Theorem was used in \cite{FGG00} to construct a quantum algorithm for optimization problems and introduced a new paradigm in quantum computing known as \emph{quantum adiabatic computing}. In the following section, we briefly explain the quantum adiabatic algorithm and use it to solve MCOs.

\section{The Quantum Adiabatic Algorithm}\label{sec:ada-algorithm}
Consider a function $f:\{0,1\}^n\to R\subseteq \real^+$ whose optimal solution $\bar{x}$ gives $f(\bar{x})=0$. Let $H_{1}$ be a Hamiltonian defined as
\begin{equation}
H_{1}=\sum_{x} f(x)\ket{x}\bra{x}.
\end{equation}
Notice that $H_1\ket{\bar{x}}=0$, and hence, $\ket{\bar{x}}$ is an eigenvector. Thus, an optimization problem reduces to finding an eigenstate with minimum eigenvalue \cite{FGG00}.
For any $s\in[0,1]$, let $H(s)=(1-s)H_0+sH_1$, where $H_0$ is an initial Hamiltonian chosen accordingly. If we initialize the system in the lowest energy eigenstate $\ket{\psi(0)}$, the adiabatic theorem guarantees that $T$ at least $1/\poly(\lambda)$ suffices to obtain a quantum state close to $\ket{\psi(1)}$, and hence, to our desired optimal solution. We call $H_1$ and $H_0$ the final and initial Hamiltonians, respectively.

After defining the initial and final Hamiltonians, the adiabatic theorem guarantees that we can find an optimal solution in finite time using the following procedure known as the \emph{Quantum Adiabatic Algorithm}. Let $H(s)=(1-s)H_0+sH_1$. Prepare the system in the ground-state $\ket{\psi(0)}$ of $H$. Then let the system evolve for time $t$ close to $T$. Finally, after time $t$, read-out the result by measuring the system in the computational basis. The only requirements, in order to make any use of the adiabatic algorithm, is that $H_0$ and $H_1$ must not commute and the total Hamiltonian $H(s)$ must be nondegenerate in its minimum eigenvalue \cite{FGG00}.

In this section we show how to construct the initial and final Hamiltonians for MCOs. Given any normal and collision-free MCO $\Pi_d^\lambda=(D,R,\mathcal{F})$ we will assume with no loss of generality that $D=\{0,1\}^n$, that is, $D$ is a set of binary words of length $n$. 

For each $i\in [1,d]_\integer$ define a Hamiltonian $H_{f_i}=\sum_{x\in\{0,1\}^n}f_i(x)\ket{x}\bra{x}$. The minimum eigenvalue of each $H_{f_i}$ is nondegenerate and 0 because $\Pi_d^\lambda$ is normal and collision-free. For any $w\in W_d$, the final Hamiltonian $H_{w}$ is defined as
\begin{align}
H_{w} &=w_1H_{f_1}+\cdots +w_dH_{f_d}\nonumber\\
    &=\sum_{x\in\{0,1\}^n} \big( w_1f_1(x)+\cdots + w_df_d(x) \big)\ket{x}\bra{x}\nonumber\\
    &=\sum_{x\in\{0,1\}^n} \langle f(x),w\rangle \ket{x}\bra{x}\label{eq:convex-comb}.
\end{align}

Following the work of \cite{FGG00}, we choose as initial Hamiltonian one that does not diagonalizes in the computational basis. Let $\ket{\hat{0}}=(\ket{0}+\ket{1})/\sqrt{2}$ and $\ket{\hat{1}}=(\ket{0}-\ket{1})/\sqrt{2}$. A quantum state $\ket{\hat{x}}$, for any $x\in \{0,1\}^n$, is obtained by applying the $n$-fold Walsh-Hadamard operation $F^{\otimes n}$ on $\ket{x}$. The set $\{\ket{\hat{x}}\}_{x\in\{0,1\}^n}$ is known as the Hadamard basis. The initial Hamiltonian is thus defined over the Hadamard basis as
\begin{equation}\label{eq:init-ham}
H_0=\sum_{x\in \{0,1\}^n} h(x)\ket{\hat{x}}\bra{\hat{x}},
\end{equation}
where $h(0^n)=0$ and $h(x)\geq 1$ for all $x\neq 0^n$. It is easy to see that the minimum eigenvalue is nondegenerate\footnote{In quantum physics, a Hamiltonian is degenerate when one of its eigenvalues has multiplicity greater than one.} with corresponding eigenstate $\ket{\hat{0}^n}=\frac{1}{\sqrt{2}}\sum_{x\in \{0,1\}^n}\ket{x}$.

For any vector $w$ in Euclidian space we define the $\ell_1$-norm of $w$ as $\|w\|_1=|w_1|+\cdots |w_d|$.

\begin{theorem}\label{the:alg}
Let $\Pi_d^\lambda$ be any normal and collision-free MCO. If there are no equivalent Pareto-optimal solutions, then for any $w\in W_d$ there exists $w'\in W_d$, satisfying $\|w-w'\|_1\leq 1/\poly(n)$, such that the quantum adiabatic algorithm, using $H_{w'}$ as final Hamiltonian, can find a Pareto-optimal solution $x$ corresponding to $w$ in finite time.
\end{theorem}

Note that if a linearization $w$ gives a nondegenerate Hamiltonian $H(s)$, we can directly use the adiabatic algorithm to find a Pareto-optimal solution. In the case of a degenerate Hamiltonian $H(s)$, Theorem \ref{the:alg} tell us that we can still find a Pareto-optimal solution using the adiabatic algorithm, provided we choose a new $w'$ sufficiently close to $w$.



\section{Eigenspectrum of the Final Hamiltonian}\label{sec:convex}
In this section we prove Theorem \ref{the:alg}. Note that if the initial Hamiltonian does not commute with the final Hamiltonian, it suffices to prove that the final Hamiltonian is nondegenerate in its minimum eigenvalue \cite{FGG00}.
For the remaining of this work, we let $\sigma_w$ and $\alpha_w$ be the smallest and second smallest eigenvalues of $H_w$ corresponding to a normal and collision-free MCO $\Pi_d^\lambda=(D,R,\mathcal{F})$.

\begin{lemma}\label{lem:min-eigenvalue}
Let $x$ be a non-trivial Pareto-optimal solution of $\Pi_d^\lambda$. For any $w\in W_d$ it holds that $\sigma_w > \langle w,\lambda\rangle$.
\end{lemma}
\begin{proof}
Let $\sigma_w=w_1f_1(x)+\dots + w_d f_d(x)$ and let $x$ be a non-trivial Pareto-optimal element. For each $w_i \in N$ we have that
\[
\sigma_w=\sum_{i} w_i f_i(x)> \sum_{i}w_i\lambda_i=\langle w,\lambda\rangle.\]
\end{proof}

\begin{lemma}\label{lem:eigen-gap}
For any $w\in W_d$, let $H_w$ be a Hamiltonian with a nondegenerate minimum eigenvalue. The eigenvalue gap between the smallest and second smallest eigenvalues of $H_w$ is at least $\langle \lambda,w\rangle$.
\end{lemma}
\begin{proof}
Let $\sigma_w$ be the unique minimum eigenvalue of $H_w$. We have that
$
\sigma_w = \langle f(x),w \rangle
$
for some $x\in \{0,1\}^n$. Now let $\alpha_w=\langle f(y),w \rangle$ be a second smallest eigenvalue of $H_w$ for some $y\in \{0,1\}^n$ where $y\neq x$. Hence,
\begin{align*}
\alpha_w-\sigma_w	&=\langle f(y),w \rangle - \langle f(x),w \rangle\\
				&=w_1f_1(y)-w_1f_1(x)+w_2f_2(y)-w_2f_2(x)\\
				&>w_1\lambda 1+w_2\lambda_2\\
				&=\langle \lambda,w\rangle.
\end{align*}
\end{proof}

\begin{lemma}\label{lem:suff-cond}
If there are no weakly-equivalent Pareto-optimal solutions in $\Pi_d^\lambda$, then the Hamiltonian $H_w$ is non-degenerate in its minimum eigenvalue.
\end{lemma}
\begin{proof}
By the contrapositive, suppose $H_w$ is degenerate in its minimum eigenvalue $\sigma_w$. Take any two degenerate minimal eigenstates $\ket{x}$ and $\ket{y}$, with $x\neq y$, such that
\[
w_1f_1(x) +\cdots+w_df_d(x)=w_1f_1(y)+\cdots +w_df_2(d)=\sigma_w.
\]
Then it holds that $x$ and $y$ are weakly-equivalent.
\end{proof}

We further show that even if $\Pi_d^\lambda$ has weakly-equivalent Pareto-optimal solutions, we can find a nondegenerate Hamiltonian. Let $m=\max_{x,i}\{f_i(x)\}$.

\begin{lemma}\label{lem:min}
For any $\Pi_d^\lambda$, let $x_1,\dots,x_\ell\in D$ be Pareto-optimal solutions that are not pairwise equivalent. If there exists $w\in W_d$ and $\sigma_w \in \real^+$ such that $\langle f(x_1),w\rangle = \cdots =\langle f(x_\ell),w\rangle=\sigma_w$ is minimum among all $y\in D$, then there exists $w'\in W_d$ and $i\in[1,\ell]_\integer$ such that for all $j\in[1,\ell]_\integer$, with $j\neq i$, it holds $\langle f(x_i),w'\rangle <\langle f(x_j),w'\rangle$. Additionally, if the linearization $w'$ satisfies  $\|w-w'\|_1 \leq \frac{\langle \lambda, w\rangle}{md}$, then $\langle f(x_i),w'\rangle$ is unique and minimum among all $\langle f(y),w'\rangle$ for $y\in D$.
\end{lemma}
\begin{proof}
We prove the lemma by induction on $\ell$. Let $\ell=2$, then $\langle f(x_1),w\rangle =\langle f(x_2),w\rangle$, and hence,
\begin{equation}
\begin{aligned}\label{eq:system-original}
w_1f_1(x_1)+\cdots +w_df_d(x_1)	&=\sigma_w\\
w_1f_1(x_2)+\cdots +w_df_d(x_2)	&=\sigma_w.
\end{aligned}
\end{equation}
for some $\sigma_w\in \real^+$. From linear algebra we know that there is an infinite number of elements of $W_d$ that simultaneously satisfy Eq.(\ref{eq:system-original}). With no loss of generality, fix $w_3,\dots,w_d$ and set $b_1=w_3f_3(x_1)+\cdots +w_df_d(x_1)$ and $b_2=w_3f_3(x_2)+\cdots +w_df_d(x_2)$. We have that
\begin{equation}\label{eq:system}
\begin{aligned}
w_1f_1(x_1) +w_2f_2(x_1)	&=\sigma_w-b_1\\
w_1f_1(x_2) +w_2f_2(x_2)	&=\sigma_w-b_2.
\end{aligned}
\end{equation}
Again, by linear algebra, we know that Eq.(\ref{eq:system}) has a unique solution  $w_1$ and $w_2$; it suffices to note that the determinant of the coefficient matrix of Eq.(\ref{eq:system}) is not 0.

Choose any $w_1'\neq w_1$ and $w_2'\neq w_2$ satisfying $w_1'+w_2'+w_3+\cdots +w_d=1$ and let $w'=(w_1',w_2',w_3,\dots,w_d)$. Then we have that $\langle f(x_1),w'\rangle \neq \langle f(x_2),w'\rangle$ because $w_1'$ and $w_2'$ are not solutions to Eq.(\ref{eq:system}). Hence, either $\langle f(x_1),w'\rangle$ or $\langle f(x_2),w'\rangle$ must be smaller than the other.

Suppose that $\langle f(x_1),w'\rangle < \langle f(x_2),w'\rangle$. We now claim that $\langle f(x_1),w'\rangle$ is mininum and unique among all $y\in D$. In addition to the constraint of the preceding paragraph that $w'$ must satisfy, in order for $\langle f(x_1),w'\rangle$ to be minimum, we must choose $w'$ such that $\| w- w'\|_1\leq \frac{\langle \lambda, w\rangle}{md}$.

Assume for the sake of contradiction the existence of $y\in D$ such that $\langle f(y),w'\rangle \leq \langle f(x_1),w'\rangle$. Hence,
\[
\langle f(y),w'\rangle \leq \langle f(x_1),w\rangle < \langle f(y),w\rangle.
\]
From Lemma \ref{lem:min-eigenvalue}, we know that $|\langle f(x_1),w\rangle - \langle f(y),w\rangle| > \langle \lambda,w\rangle$, and thus,
\begin{equation}\label{eq:contradiction}
|\langle f(y),w'\rangle - \langle f(y),w\rangle| > \langle \lambda,w\rangle.
\end{equation}
Using the Cauchy-Schwarz inequality we have that
\begin{align*}
|\langle f(y),w'\rangle - \langle f(y),w\rangle|	&=|\langle f(y),w'-w\rangle|\\
									&\leq \|f(y)\|_1 \cdot \|w'-w\|_1\\
									&\leq \langle \lambda,w\rangle,
\end{align*}
where the last line follows from $\|f(y)\|_1\leq md$ and $\| w- w'\|_1\leq \frac{\langle \lambda, w\rangle}{md}$; from Eq.(\ref{eq:contradiction}), however, we have that $|\langle f(y),w'-w\rangle|>\langle \lambda,w\rangle$, which is a contradiction. Therefore, we conclude that $\langle f(y),w'\rangle > \langle f(x_1),w'\rangle$ for any $y\in D$; the case for $\langle f(x_1),w'\rangle > \langle f(x_2),w'\rangle$ can be proved similarly. The base case of the induction is thus proved.

Now suppose the statement holds for $\ell$. Let $x_1,\dots,x_\ell,x_{\ell+1}$ be Pareto-optimal solutions that are not pairwise equivalent. Let $w\in W_d$ be such that $\langle f(x_1),w\rangle = \cdots =\langle f(w_{\ell+1}),w\rangle$ holds. By our induction hypothesis, there exists $w'\in W_d$ and $i\in [1,\ell]_\integer$ such that $\langle f(x_i),w'\rangle  < \langle f(y),w' \rangle$ for any other $y\in D$.

If $ \langle f(x_i),w'\rangle \neq \langle f(x_{\ell+1}),w'\rangle$ then we are done, because either one must be smaller. Suppose, however, that $\langle f(x_{\ell+1}),w'\rangle= \langle f(x_i),w'\rangle=\sigma_{w'}$ for some $\sigma_{w'} \in \real^+$. From the base case of the induction we know there exists $w''\neq w'$ that makes $ \langle f(x_i),w''\rangle < \langle f(x_{\ell+1}),w''\rangle$, and hence, $\langle f(x_i),w''\rangle < \langle f(y),w''\rangle$ for any $y\in D$.
\end{proof}

The premise in Lemma \ref{lem:min}, that each $x_1,\dots,x_\ell$ must be Pareto-optimal solutions, is a sufficient condition because if one solution is not Pareto-optimal, then the statement will contradict Lemma \ref{lem:sup}.

We now apply Lemma \ref{lem:min} to find a Hamiltonian with a nondegenerate minimum eigenvalue.

\begin{lemma}\label{lem:min-eigen}
Let $\Pi_d^\lambda$ be a MCO with no equivalent Pareto-optimal solutions and let $H_w$ be a degenerate Hamiltonian in its minimum eigenvalue with corresponding minimum eigenstates $\ket{x_1},\dots,\ket{x_\ell}$. There exists $w'\in W_d$, satisfying $\|w-w'\|_1\leq \frac{\langle \lambda,w\rangle}{md}$, and $i\in[1,\ell]_\integer$ such that $H_{w'}$ is nondegenerate in its smallest eigenvalue with corresponding eigenvector $\ket{x_i}$.
\end{lemma}
\begin{proof}
From Lemma \ref{lem:suff-cond}, we know that if $\Pi_d^\lambda$ has no weakly-equivalent Pareto-optimal solutions, then for any $w$ the Hamiltonian $H_w$ is nondegenerate.

We consider now the case when the minimum eigenvalue of $H_w$ is degenerate with $\ell$ Pareto-optimal solutions that are weakly-equivalent. Let  $x_1,\dots,x_{\ell}$ be such weakly-equivalent Pareto-optimal solutions that are non-trivial and $x_i \not\equiv x_j$ for all $i\neq j$. By Lemma \ref{lem:min} there exists $w'\in W_d$, where $w\neq w'$, such that $\langle f(x_i),w'\rangle$ is minimum among all $y \in D$.
\end{proof}

If we consider our assumption from Section \ref{sec:MCO} that $m=\poly(n)$, where $n$ is the maximum number of bits of any element in $D$, we have that any $w'$ must satisfy $\|w-w'\|_1\leq 1/\poly(n)$.  Then Theorem \ref{the:alg} follows immediately from lemmas \ref{lem:sup} and \ref{lem:min-eigen}.

To see that the adiabatic evolution takes finite-time let $\Delta_{max}=\max_s \| \frac{d}{ds}H(s) \|$ and $g_{min}=\min_s g(s)$, where $g(s)$ is the eigenvalue gap of $H(s)$. Letting $T=O(\frac{\Delta_{max}}{g_{min}^2})$ suffices to find a supported solution corresponding to $w$. Since $g_{min}>0$ and $\| \frac{d}{ds}H(s) \|=\poly(n)$, we conclude that $T$ is finite.

\section{Application of the Adiabatic Algorithm to the Two-Parabolas Problem}\label{sec:two-parabolas}

To make use of the adiabatic algorithm of Section \ref{sec:ada-algorithm} in the Two-Parabolas problem we need to consider a collision-free version of the problem. Let $TP_2^\lambda=(D,R,\mathcal{F})$ be a normal and collision-free MCO where $\lambda=(\lambda_1,\lambda_2)\in \real^+ \times \real^+$, $D=\{0,1\}^n$, $R\subseteq \real^+$ and $\mathcal{F}=\{f_1,f_2\}$. Let $x_0$ and $x_0'$ be the optimal solutions of $f_1$ and $f_2$, respectively. We will use $x_i$ to indicate the $i$th solution of $f_1$ and $x_i'$ for $f_2$. Moreover, we assume that $|x_0 -x_0'| > 1$. This latter assumption will ensure that there is at least one non-trivial Pareto-optimal solution.

To make $TP_2^\lambda$ a Two-Parabolas problem, we impose the following conditions.
\begin{enumerate}
\item For each $x\in [0,x_0]$, the functions $f_1$ and $f_2$ are decreasing;
\item for each $x\in [x_0',2^n-1]$, the functions $f_1$ and $f_2$ are increasing;
\item for each $x\in [x_0+1,x_0'-1]$ , the function $f_1$ is increasing and the function $f_2$ is decreasing.
\end{enumerate}

The final and initial Hamiltonians are as in Eq.(\ref{eq:convex-comb}) and Eq.(\ref{eq:init-ham}), respectively. In particular, in Eq.(\ref{eq:init-ham}), we define the initial Hamiltonian as
\begin{equation}\label{eq:tp-h0}
\hat{H}_0=\sum_{x\in\{0,1\}^n\setminus \{0^n\}} \ket{\hat{x}}\bra{\hat{x}}.
\end{equation}
Thus, the Hamiltonian of the entire system for $TP_2^\lambda$ is
\begin{equation}\label{eq:ham-tp2}
H(s)=(1-s)\hat{H}_0+sH_{w}.
\end{equation}

From the previous section we know that $T=O(\frac{\Delta_{max}}{g_{min}^2})$ suffices to find a supported solution corresponding to $w$ \cite{DMV01}. The quantity $\Delta_{max}$ is usually easy to estimate. The eigenvalue gap $g_{min}$ is, however, very difficult to compute; indeed, determining for any  Hamiltonian if $g_{min}>0$ is undecidable \cite{CPW15}.

We present a concrete example of the Two-Parabolas problem on six qubits and numerically estimate the eigenvalue gap. In Fig.\ref{fig:tp-six} we show a discretized instance of the Two-Parabolas problem---Table \ref{tab:tp-spec} presents a complete specification of all points.

\begin{figure}
\centering
\includegraphics[scale=0.80]{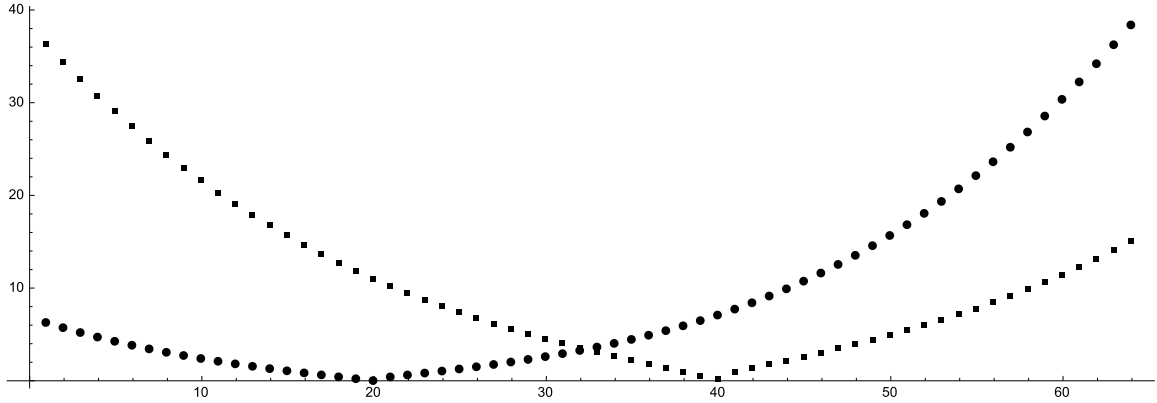}
\caption{A discrete Two-Parabolas problem on seven qubits. Each objective function $f_1$ and $f_2$ is represented by the rounded points and the squared points, respectively. The gap vector $\lambda=(0.2,0.4)$. The trivial Pareto-optimal points are 40 and 80.}
\label{fig:tp-six}
\end{figure}

\begin{table}
\centering
\caption{Complete definition of the Two-Parabolas example of Fig.\ref{fig:tp-six} for seven qubits.}
\resizebox{0.9\textwidth}{!}{\begin{minipage}{\textwidth}
\label{tab:tp-spec}
\begin{tabular}{||c|c|c||c|c|c||c|c|c||c|c|c||c}
\hline\hline
$x$	& $f_1(x)$		& $f_2(x)$ 	& $x$	& $f_1(x)$		& $f_2(x)$&
$x$	& $f_1(x)$		& $f_2(x)$ 	& $x$	& $f_1(x)$		& $f_2(x)$\\
\hline
1	& 36.14		& 214.879&
2	& 34.219		& 208.038&	
3	& 32.375		& 201.354&
4	& 30.606		& 194.825\\
5	& 28.91		& 188.449&
6	& 27.285		& 182.224&
7	& 25.729		& 176.148&
8	& 24.24		& 170.219\\
9	& 22.816		& 164.435&
10	& 21.455		& 158.794&
11	& 20.155		& 153.294&
12	& 18.914		& 147.933\\
13	& 17.73		& 142.709&
14	& 16.601		& 137.62&
15	& 15.525		& 132.664&
16	& 14.5		& 127.839\\
17	& 13.524		& 123.143&
18	& 12.595		& 118.574&
19	& 11.711		& 114.13&
20	& 10.87		& 109.809\\
21	& 10.07		& 105.609&
22	& 9.309		& 101.528&
23	& 8.585		& 97.564&
24	& 7.896		& 93.715\\
25	& 7.24 		& 89.979&
26	& 6.615		& 86.354&
27	& 6.019		& 82.838&
28	& 5.45		& 79.429\\
29	& 4.906		& 76.125&
30	& 4.385		& 72.924&
31	& 3.885		& 69.824&
32	& 3.404		& 66.823\\
33	& 2.94		& 63.919&
34	& 2.491		& 61.11&
35	& 2.055		& 58.394&
36	& 1.63		& 55.769\\
37	& 1.214		& 53.233&
38	& 0.805		& 50.784&
39	& 0.401		& 48.42&
40	& 0			& 46.139\\
41	& 0.801		& 43.939&
42	& 1.205		& 41.818&
43	& 1.614		& 39.774&
44	& 2.03		& 37.805\\
45	& 2.455		& 35.909&
46	& 2.891		& 34.084&
47	& 3.34		& 32.328&
48	& 3.804		& 30.639\\
49	& 4.285		& 29.015&
50	& 4.785		& 27.454&
51	& 5.306		& 25.954&
52	& 5.85		& 24.513\\
53	& 6.419		& 23.129&
54	& 7.015		& 21.8&
55	& 7.64		& 20.524&
56	& 8.296		& 19.299\\
57	& 8.985		& 18.123&
58	& 9.709		& 16.994&
59	& 10.47		& 15.91&
60	& 11.27		& 14.869\\
61	& 12.111		& 13.869&
62	& 12.995		& 12.908&
63	& 13.924		& 11.984&
64	& 14.9		& 11.095\\
65	& 15.925		& 10.239&
66	& 17.001		& 9.414	&
67	& 18.13		& 8.618	&
68	& 19.314		& 7.849	\\
69	& 20.555		& 7.105	&
70	& 21.855		& 6.384	&
71	& 23.216		& 5.684	&
72	& 24.64		& 5.003	\\
73	& 26.129		& 4.339	&
74	& 27.685		& 3.69	&
75	& 29.31		& 3.054	&
76	& 31.006		& 2.429	\\
77	& 32.775		& 1.813	&
78	& 34.619		& 1.204	&
79	& 36.54		& 0.6	&
80	& 38.54		& 0	\\
81	& 40.621		& 1.2	&
82	& 42.785		& 1.804	&
83	& 45.034		& 2.413	&
84	& 47.37		& 3.029	\\
85	& 49.795		& 3.654	&
86	& 52.311		& 4.29	&
87	& 54.92		& 4.939	&
88	& 57.624		& 5.603	\\
89	& 60.425		& 6.284	&
90	& 63.325		& 6.984	&
91	& 66.326		& 7.705	&
92	& 69.43		& 8.449	\\
93	& 72.639		& 9.218	&
94	& 75.955		& 10.014	&
95	& 79.38		& 10.839	&
96	& 82.916		& 11.695	\\
97	& 86.565		& 12.584	&
98	& 90.329		& 13.508	&
99	& 94.21		& 14.469	&
100	& 98.21		& 15.469	\\
101	& 102.331		& 16.51	&
102	& 106.575		& 17.594	&
103	& 110.944		& 18.723	&
104	& 115.44		& 19.899	\\
105	& 120.065		& 21.124	&
106	& 124.821		& 22.4	&
107	& 129.71		& 23.729	&
108	& 134.734		& 25.113	\\
109	& 139.895		& 26.554	&
110	& 145.195		& 28.054	&
111	& 150.636		& 29.615	&
112	& 156.22		& 31.239	\\
113	& 161.949		& 32.928	&
114	& 167.825		& 34.684	&
115	& 173.85		& 36.509	&
116	& 180.026		& 38.405	\\
117	& 186.355		& 40.374	&
118	& 192.839		& 42.418	&
119	& 199.48		& 44.539	& 
120	& 206.28		& 46.739	\\
121	& 213.241		& 49.02	&
122	& 220.365		& 51.384	&
123	& 227.654		& 53.833	&
124	& 235.11		& 56.369	\\
125	& 242.735		& 58.994	&
126	& 250.531		& 61.71	&
127	& 258.5		& 64.519	&
128	& 266.644		& 67.423	\\
\hline\hline
\end{tabular} 
\end{minipage}}
\end{table}

For this particular example we use as initial Hamiltonian $8H_0$, that is, Eq.(\ref{eq:tp-h0}) multiplied by 8. Thus, the minimum eigenvalue of $8H_0$ is 0, whereas any other eigenvalue is 8.  

In Fig.\ref{fig:eigengap} we present the eigenvalue gap of $TP_2^\lambda$ for $w=0.57$ where we let $w_1=w$ and $w_2=1-w_1$; for this particular value of $w$ the Hamiltonian $H_{F,w}$ has a unique minimum eigenstate which corresponds to Pareto-optimal solution 59. The two smallest eigenvalues never touch, and exactly at $s=1$ the gap is $|\langle w,f(x_0)\rangle-\langle w,f(x_1)\rangle|$, where $x_0=59$ and $x_1=60$ are the smallest and second smallest solutions with respect to $w$, which agrees with lemmas \ref{lem:min-eigenvalue} and \ref{lem:eigen-gap}.

Similar results can be observed for different values of $w$ and a different number of qubits. Therefore, the experimental evidence lead us to conjecture that in the Two-Parabolas problem $g_{min}\geq  |\langle w, f(x)\rangle-\langle w,f(y)\rangle|$, where $x$ and $y$ are the smallest and second smallest solutions with respect to $w$.


\begin{figure*}
\centering
\includegraphics[scale=0.9]{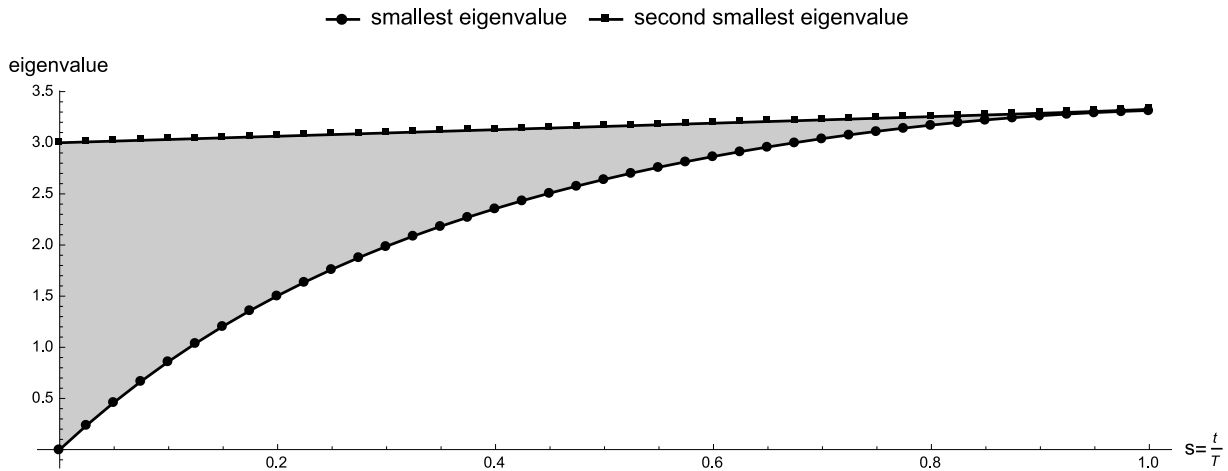}
\caption{Eigenvalue gap (in gray) of the Two-Parabolas problem of Fig.\ref{fig:tp-six} for $w=0.57$. The eigenvalue gap at $s=1$ is exactly $|\braket{w}{f(x)}-\braket{w}{f(y)}|$, where $x=59$ and $y=60$ are the smallest and second smallest solutions with respect to $w$.}
\label{fig:eigengap}
\end{figure*}

\section{Concluding Remarks and Open Problems}\label{sec:open}
In this work we showed that the quantum adiabatic algorithm of \cite{FGG00} can be used for multiobjective combinatorial optimization problems. In particular, a simple linearization of the objective functions suffices to guarantee convergence to a Pareto-optimal solution provided the linearized single-objective problem has an unique optimal solution. Nevertheless, even if a linearization of objectives does not give an unique optimal solution, then it is always possible to choose an appropriate linearization that does.

We end this paper by listing a few promising and challenging open problems.
\begin{enumerate}
\item To make any practical use of Theorem \ref{the:alg} we need to chose $w\in W_d$ in such a way that the optimal solution of the linearization of an MCO  has an unique solution. It is very difficult, however, to know a priori which $w$ to chose in order to use the adiabatic algorithm. Therefore, more research is necessary to learn how to select these linearizations. One way could be to constraint the domain of an MCO in order to minimize the number of weak-equivalent solutions.
\item Another related issue is learn how to solve multiobjective problems in the presence of equivalent solutions. A technique of mapping an MCO with equivalent solutions to Hamiltonians seems very difficult owing to the fact that the smallest eigenvalue must be unique in order to apply the adiabatic theorem.
\item According to Theorem \ref{the:alg}, we can only find all supported solutions. Other works showed that the number of non-supported solutions can be much larger than the number of supported solutions \cite{EG00}. Hence, it is interesting to construct a quantum algorithm that could find an approximation to all Pareto-optimal solutions.
\item Prove our conjecture of Section \ref{sec:two-parabolas} that the eigenvalue gap of the Hamiltonian of Eq.(\ref{eq:ham-tp2}), corresponding to the Two-Parabolas problem, is at least the difference between the smallest solution and second smallest solution for any given linearization of the objective functions.
\end{enumerate}

\bibliographystyle{alpha}
\bibliography{../../../library}

\end{document}